\documentclass[runningheads]{llncs}
\usepackage{amsmath}
\usepackage{xspace}
\usepackage{amssymb}
\usepackage{mathrsfs}
\usepackage[T1]{fontenc}
\usepackage{algorithm}
\usepackage{algpseudocode}
\usepackage{todonotes}
\usepackage{mathdots}
\usepackage{pgfplots}
\usepackage{booktabs}
\usepackage{subfig}
\pgfplotsset{compat=1.8}
\usepackage{pgfplotstable}
\usepackage{hyperref}
\usepackage[detect-all]{siunitx}
\usepackage{etoolbox}
\usetikzlibrary{spy}
\pgfplotstableset{col sep=semicolon}

\newrobustcmd\B{\DeclareFontSeriesDefault[rm]{bf}{b}\bfseries}

\newcommand{\fodot}{FO($\cdot$)\xspace}

\title{Efficiently grounding FOL using bit vectors}
\author{Lucas Van Laer, Simon Vandevelde, Joost Vennekens}

\institute
{KU Leuven, De Nayer Campus, Dept. Of Computer Science,\\
J.-P. De Nayerlaan 5, 2860 Sint-Katelijne-Waver, Belgium
\and
Leuven.AI -- KU Leuven institute for AI, B-3000 Leuven, Belgium
\and
Flanders Make -- DTAI-FET\\
\email{\{lucas.vanlaer, s.vandevelde, joost.vennekens\}@kuleuven.be}
}

\authorrunning{L. Van Laer et al.}

\begin{document}
\maketitle

\begin{abstract}
Several paradigms for declarative problem solving start from a specification in a high-level language, which is then transformed to a low-level language, such as SAT or SMT.
Often, this transformation includes a ``grounding'' step to remove first-order quantification.
To reduce the time and size of the grounding, it can be useful to simplify formulas along the way, e.g., by already taking into account the interpretation of symbols that are already known. 
In this paper, we investigate the use of bit vectors to efficiently simplify formulas, thereby taking advantage of the fact that, on modern hardware, logical operations on bit vectors can be executed extremely fast.
We conduct an experimental analysis, which shows that bit vectors are indeed fast for certain problems, but also have limitations.
\end{abstract}

\section{Introduction}

Many systems for declarative problem solving, including, for instance, most Answer Set Programming (ASP)~\cite{ASP} systems, use a ``ground and solve'' approach: the user represents the problem domain in a ``high-level'' specification language, which is then transformed into an equivalent ``low-level'' specification that can be used to efficiently search for solutions.
Typically, the high-level language makes use of first-order variables to represent the domain knowledge in a way which is independent of the size of the problem that needs to be solved: e.g., the rules of an $n$-queens problem can be represented in a way which does not depend on the value of $n$.
For instance, in classical logic syntax, one of the relevant constraints can be written as:

\begin{equation}\forall x,y: x\neq y \Rightarrow \mathit{queen}(x) \neq \mathit{queen}(y).\label{eq:fo}\end{equation}

In order to actually solve a concrete instance of the problem, say for $n=8$, the variables $x$ and $y$ are then grounded out, i.e., replaced by all possible combinations in $[1,8] \times [1,8]$. This step turns the domain-size independent formula into one which does depend on the chosen size of the domain:

\[ (1\neq1 \Rightarrow \mathit{queen}(1) \neq \mathit{queen}(1)) \land (1\neq 2 \Rightarrow \mathit{queen}(1) \neq \mathit{queen}(2)) \land  \ldots \] 

However, this formula can be further simplified, e.g., into:

\begin{equation}\mathit{queen}(1) \neq \mathit{queen}(2)) \land  \mathit{queen}(1) \neq \mathit{queen}(3)) \land \ldots\label{eq:gr}\end{equation}


This simplified formula can then be given to a solver. 
For instance, the ASP system clingo~\cite{clingo} sends the output of its grounder \textit{gringo} to its solver \textit{clasp}.
Alternatively, a ground formula that has been appropriately simplified can also be sent to a SAT or SMT solver.
For instance, the IDP-Z3 system~\cite{IDP-Z3} grounds formulas in the \fodot language~\cite{FOdot}, a rich extension of classical logic, to SMT and then sends these to the Z3 solver~\cite{Z3}.
Its predecessor, the IDP3 system~\cite{IDP3}, uses the same input language and grounds it to an extended SAT format which is sent to a custom variant of the minisat SAT-solver~\cite{MinisatID}.

This paper is concerned with the question of how to efficiently generate a suitable simplified grounding, i.e., how to produce \eqref{eq:gr} from \eqref{eq:fo}.
In the above example, this is done based on the desired size of the domain ($n=8$) and the intended interpretation of the equality symbol (e.g., $1=1$ and $1 \neq 2$).
In general, this step may also take into account specific information about the problem instance. Consider for instance a map colouring problem:

\begin{equation} \forall x,y: x \neq y \land \mathit{border}(x,y) \Rightarrow \mathit{colour}(x) \neq \mathit{colour}(y). \end{equation}

Here, a problem instance  describes not only how many countries there are, but also the borders between them. In the simplification process, this information allows to reduce the above formula to a grounding such as:

\begin{equation}
    \mathit{colour}(\mathit{USA}) \neq \mathit{colour}(\mathit{Canada}) \land \mathit{colour}(\mathit{USA}) \neq \mathit{colour}(\mathit{Mexico}) \land \cdots
\end{equation}

In this paper, we present a method which uses bit vectors to carry out such reductions in an efficient way.
We begin by introducing some preliminaries in Section~\ref{s:preliminary}.
Next, we elaborate on our theoretical approach in Section~\ref{s:theory}, and we demonstrate a concrete implementation for formulas in which all symbols are interpreted (Section~\ref{s:implementation}), and for formulas with a mix of (non-)interpreted symbols (Section~\ref{s:moreimplementation}).
Lastly, we discuss relevant related work in Section~\ref{s:relatedwork}, present an evaluation in Section~\ref{s:evaluation}, and conclude in Section~\ref{s:conclusion}.



\section{Preliminaries}\label{s:preliminary}





We consider classical first-order logic (FO). A \emph{vocabulary} $\Sigma$ is a set of predicate and function symbols, each of which has an associated arity $n\geq 0$.  A \emph{structure} for a vocabulary $\Sigma$ has a domain $D$ and assigns to each predicate symbol $P/n$ in $\Sigma$ a $n$-ary relation on $D$ (i.e., $P^S \subseteq D^n$) and to each function symbol $f/n$ an $n$-ary function on $D$ (i.e., $f^S$ is a function $D^n\rightarrow D$).  In this paper, we only consider structures with a finite domain $D$, which ensures that groundings are finite.

A \emph{term} is either a variable or a function symbol $f/n$ applied to $n$ terms. An \emph{atom} is a predicate symbol $P/n$ applied to $n$ terms. Formulas are constructed in the usual way by combining atoms with the logical operators $\forall, \exists, \lnot, \land, \lor$.  Implication $\phi \Rightarrow \psi$ and equivalence $\phi \Leftrightarrow \psi$ are abbreviations for $\lnot \phi \lor \psi$ and $(\phi \Rightarrow \psi)\land (\phi \Leftarrow \psi)$, respectively.  The satisfaction relation $S \models \phi$ between structures $S$ and formulas $\phi$ is also defined as usual.

In this paper, we assume a typed logic, in which each vocabulary designates a number of unary predicates $\tau_1,\ldots,\tau_n$ as types. We restrict attention to structures $S$ such that the domain of $\mathit{dom}(S)$ is partitioned by $\tau^S_1,\ldots,\tau^S_n$, i.e., each object in the domain belongs to exactly one type.
Whenever we use a variable $x$, we assume that we know the type over which $x$ ranges and denote this as $\tau_x$.
In practice, this can be achieved by asking the modeller to explicitly indicate the type of each variable, or by performing type inference.

It is useful to consider the grounding problem in the context of  Model Expansion (MX) \cite{Mitchell2006}:  we consider a formula $\phi$ in vocabulary $\Sigma$ and a structure $S_0$ for \emph{part of} $\Sigma$, i.e., if  $\Sigma_0$ is the set of symbols interpreted by structure $S_0$, we have that $\Sigma_0 \subseteq \Sigma$.
In the map colouring example above, $\Sigma_0 = \{ \mathit{border}/2, ={}/2\}$. We then consider the problem of
finding a structure $S_1$ for the vocabulary $\Sigma \setminus \Sigma_0$ (= $\{ \mathit{colour} \}$ in
the example) such that $S_0 \cup S_1 \models \phi$. Here, $S_1$ must be a structure that has the
same domain as $S_0$, so that the union $S_0 \cup S_1$ is simply the structure $S$ in which,
for each symbol $\sigma \in \Sigma$, if $\sigma \in \Sigma_0$ then $\sigma^S = \sigma^{S_0}$ and
if $\sigma \not \in \Sigma_0$, then $\sigma^S = \sigma{^S_1}$.
The set of all such structures $S_1$ is denoted as $MX(\phi,S_0)$, i.e.,
the \emph{model expansions} of $S_0$ w.r.t.~$\phi$.

In this context, we can view the goal of the grounding process as producing a propositional formula $\psi$ such that
$S_1 \models \psi$ if and only if $S_1 \in MX(\phi, S_0)$. 
Here, $\psi$ is a formula in vocabulary $\Sigma$, augmented with
constant symbols that represent the elements in
the domain of structure $S_0$ (such as the
constants \emph{USA}, \emph{Canada} and \emph{Mexico} above).  While $\psi$ could in principle still contain symbols from $\Sigma_0$, the purpose of the simplification process we consider in this paper would be to reduce the formula so that only symbols from $\Sigma\setminus \Sigma_0$ remain.

\section{Computing Satisfying Sets} \label{s:theory}

The key component of our approach is a method to compute the \emph{satisfying set} of a formula in a given structure, which we define as follows.
Let $\phi$ be a formula and let $\vec{x} = (x_1,\dots, x_n)$ be a tuple of variables that contains at least all the free variables of $\phi$.
We denote the pair $(\phi, \vec{x})$ as $\phi :: \vec{x}$. 
For such a $\phi :: \vec{x}$ and a structure $S$ that interprets all of the symbols in $\phi$,
the \emph{satisfying set} of $\phi :: \vec{x}$ in $S$ is the set of all $n$-tuples of domain elements $\vec{d} \in \tau_{x_1}^S \times \ldots \times \tau_{x_n}^S$ for which $S \models \phi[\vec{x} / \vec{d}]$.
In other words, when we interpret each variable $x_i$ as $d_i$, then $S$ satisfies $\phi$.
We denote the satisfying set of $\phi :: \vec{x}$ in $S$ by $[\phi :: \vec{x}]_S$.

A naive way of computing a satisfying set is by individually checking the value of
$\phi$ for each tuple of domain elements $\vec{d}$:

\begin{algorithm}
\caption{Naive computation of satisfying set $[\phi::\vec{x}]_S$}
\label{alg:naive}
\begin{algorithmic}

\State $satset \gets \{\}$
\For{$d_1 \in \tau_{x_1}^S$}
\State $\ddots$
\For{$d_n \in \tau_{x_n}^S$}
    \If{$S[x_1/d_1,\ldots,x_n/d_n] \models \phi$}
        \State $satset.add(\vec{d})$
    \EndIf
\EndFor
\State $\iddots$
\EndFor
\end{algorithmic}
\end{algorithm}

Since the nested for-loops are $O(k^n)$, with $k$ the maximal $\lvert\tau_{x_i}^S\rvert$,  Algorithm~\ref{alg:naive} can be quite slow for formulas
with many quantifiers.
As an alternative to Algorithm~\ref{alg:naive}, we can use a recursive approach to compute satisfying sets, as demonstrated in the following proposition.


\begin{proposition}
\label{prop:satsetops}
    Let $\phi :: \vec{x}$ be a formula and S a structure that interprets all of its symbols.
    \begin{itemize}
        \item If $\phi$ is a conjunction $\psi_1 \land \psi_2$ then $[\phi :: \vec{x}]_S$ is $[\psi_1 :: \vec{x}]_S \cap {[\psi_2 :: \vec{x}]_S}$
        \item If $\phi$ is a disjunction $\psi_1  \lor \psi_2 $ then $[\phi :: \vec{x}]_S$ is $[\psi_1 :: \vec{x}]_S \cup {[\psi_2 :: \vec{x}]_S}$
        \item If $\phi$ is a negation $\lnot \psi$ then $[\phi :: \vec{x}]_S$ is $\tau_{x_1}^S \times \ldots \times \tau_{x_n}^S \setminus [\psi :: \vec{x}]_S$
        \item As a base case, for an atom $A$ of form $P(t_1, \ldots, t_m)$, $[A::\vec{x}]_S$ is the set of all $n$-tuples $\vec{d}$ such that $\vec{t}^{S[\vec{x}/\vec{d}]} \in P^S$.
    \end{itemize}
\end{proposition}
\begin{proof}
    Follows immediately from the definition of $\models$.
\end{proof}

\begin{example}
    Let $\Sigma$ be a vocabulary with type $T$ and predicates $P/2$ and $Q/2$, for which both arguments are of type $T$.
    Let $S$ be a structure that contains the following interpretations: $T = \{ a, b \}$, $P = \{ (a, a), (b, a) \}$, and $Q = \{ (a, b), (a, a) \}$.
    Let $F = P(x, y) \land \lnot Q(x, y)$.
    Using Proposition~\ref{prop:satsetops}, we can recursively calculate $[F::(x,y)]_S$. 
    For $P(x,y)$, the satisfying set of $[P(x,y)::(x,y)]_S$ is $\{ (a, a), (b, a) \}$.    The satisfying set of the second conjunct $\lnot Q(x,y)$ is $T^S \times T^S \setminus [Q :: (x,y)]_S = \{(b,a), (b,b)\}$. 
   Thus, $[F :: (x, y)]_S$ is $[P(x,y) :: (x,y)]_S \cap [\lnot Q(x,y) :: (x,y)]_S = \{(b, a)\}$.
\end{example}

A limitation of the results of Proposition~\ref{prop:satsetops} is that they only allow the satisfying set of $\phi::\vec{x}$ to be computed from the satisfying sets of $\phi$'s subformulas if these have precisely the same variables, i.e., if they are also of the form $\psi :: \vec{x}$.
Therefore, we cannot directly use these results to compute, e.g., the satisfying set of a formula $P(x) \land Q(y)$.
The following propositions introduce two ways of manipulating the variables of a formula, by either adding a free variable or swapping the position of two free variables.

\begin{proposition}
\label{prop:varadd}
    Let $\phi :: \vec{x}$ be a formula and $S$ a structure that interprets all of its symbols.
    For a variable $y \not \in \vec{x}$, let us consider the tuple $\vec{x'}$ of variables $(x_1, \dots, x_i, y, x_{i+1}, \ldots, x_n)$, for some $i \in 1..n$. Then $[\phi :: \vec{x'}]_S =\\ \{(d_1, \ldots, d_i, e, d_{i+1},\ldots, d_n) \mid (d1, \ldots, d_n) \in [\phi :: \vec{x}]_S$ and $e \in \tau_y^S\}$.
\end{proposition}

\begin{proposition}
\label{prop:varswap}
    Let $\phi :: \vec{x}$ be a formula and $S$ a structure that interprets all of its symbols.
    For $i,j \in 1..n$, let $\vec{x'}$ be the tuple of variables in which $x_i$ and $x_j$ have been swapped,
    i.e., if $\vec{x} = (\ldots, x_i, \ldots, x_j, \ldots)$ then $\vec{x'} = (\ldots, x_j, \ldots, x_i,\ldots)$.
    Here, $[\phi :: \vec{x'}]_S = \{(\ldots, d_j, \ldots, d_i, \ldots) \mid (\ldots, d_i, \ldots, d_j, \ldots) \in [\phi :: \vec{x}]_S\}$
\end{proposition}


Finally, we also need to consider the case of quantifications.

\begin{proposition}
\label{prop:quant}
    Let $\phi :: \vec{x}$ be a formula and $S$ a structure that interprets all of its symbols.
    \begin{itemize}
        \item if $\phi :: (\ldots, x_{i-1}, x_{i+1}, \ldots)$ is a universally quantified formula $\forall x_i: (\psi :: \vec{x})$, then $[\phi :: (\ldots, x_{i-1}, x_{i+1}, \ldots)]_S =$ \\$\{ (\ldots, d_{i-1}, d_{i+1}, \ldots) \mid \forall d: (\ldots, d_{i-1}, d, d_{i+1}, \ldots) \in [\phi :: \vec{x}]_S\}$
        \item if $\phi :: (\ldots, x_{i-1}, x_{i+1}, \ldots)$ is an existentially quantified formula $\exists x_i: (\psi :: \vec{x})$, then $[\phi :: (\ldots, x_{i-1}, x_{i+1}, \ldots)]_S =$ \\$ \{ (\ldots, d_{i-1}, d_{i+1}, \ldots) \mid \exists d: (\ldots, d_{i-1}, d, d_{i+1}, \ldots) \in [\phi :: \vec{x}]_S\}$
    \end{itemize}
\end{proposition}

\section{Implementation using bit vectors}\label{s:implementation}
\label{bvimpl}

The previous section shows that we can recursively derive the satisfying set of a formula from the satisfying sets of its subformulas. To do this in an efficient way, however, a good representation for these sets is needed. In this section, we will describe how we can build an efficient implementation by using bit vectors to represent  satisfying sets.  Key to this approach is that the \emph{Single Instruction, Multiple Data (SIMD)} parallelism of  modern processors can execute logical operations on large bit vectors very efficiently. Representing a satisfying set by a vector can be done as follows.

To represent a satisfying set for a formula $\phi::(x)$ with a single free variable $x$, it suffices to fix an order on the domain elements $d_1, \ldots, d_n \in \tau^S_{x}$.  We can then represent $[\phi::(x)]_S$ by a bit vector $(b_1, \ldots, b_n)$ of size $n$, with each $b_i \in \{0,1\}$: for each $d_i \in \tau_x^S$, $b_i = 1$ iff $d_i \in [\phi::(x)]_S$.

Once we have fixed an ordering on the interpretation $\tau_i^S$ of each type $\tau_i$, we can derive an ordering on tuples in $\prod_{i \in 1..n} \tau_{x_i}^S$, e.g., by means of a lexicographical ordering.  This then allows us to also represent satisfying sets for formulas $\phi::(x_1, \ldots, x_n)$ with more than one variable. Note that this requires a bit vector of size $\prod_{i \in 1..n} m_i$, where each $m_i$ is the size of $\tau_{x_i}^S$.

Following Proposition~\ref{prop:satsetops}, we can apply bitwise and/or/not-operations to combine the satisfying sets of $\psi_1::\vec{x}$ and $\psi_2::\vec{x}$ into a satisfying set for $\psi_1\land \psi_2, \psi_1\lor \psi_2$ and $\lnot \psi_1$, respectively.




\begin{example}
Let $F::(x)$ be the formula $p(x) \lor q(x)$ and let $S$ interpret $\tau_x^S = \{ a,b,c\}$, $p^S = \{b\}$ and  $q^S=\{b,c\}$.  Using the order $a < b < c$, we can calculate the satisfying set $[F::(x)]_S = \{b,c\}$ of $F$ as:

\[
\left[ \begin{array}{c c c} 0 & 1 & 0 \end{array} \right]
\lor 
\left[ \begin{array}{c c c} 0 & 1 & 1 \end{array} \right]
=
\left[ \begin{array}{c c c} 0 & 1 & 1 \end{array} \right]
\]
\end{example}

Following Proposition~\ref{prop:varadd}, starting from the satisfying set for $\phi::\vec{x}$, we can compute the satisfying set for $[\phi::\vec{x} \cup \{y\}]$ by essentially constructing the Cartesian product of $[\phi::\vec{x}]_S$ and $\tau_y^S$.  This boils down to copying the original bit vector $\lvert \tau_y^S\rvert$ times.  Continuing on the above example, we use $[q(x)::(x)]_S$ to construct $[q(x)::(x,y)]_S$, where we depict the second bit vector in a 2-dimensional layout for clarity:

\begin{minipage}{.44\linewidth}
\begin{equation*}
\label{eq:addvar}
\begin{array}{c c}
    
q(x) :: (x) \rightarrow
\begin{array}{c}
    \left[
    \begin{array}{c c c}
        0 & 1 & 1 
    \end{array}
\right]
\end{array} 
\end{array}
\end{equation*}
\end{minipage}
\begin{minipage}{.49\linewidth}
\begin{equation*}
q(x) :: (x\; y) \rightarrow
\begin{array}{c c c}
    & & x \\
    & & \begin{array}{c c c}
        a & b & c\\
    \end{array}\\
    y &
    \begin{array}{c}
        a \\
        b \\
        c \\
    \end{array} & \left[
    \begin{array}{c c c}
    0 & 1 & 1  \\
    0 & 1 & 1  \\
    0 & 1 & 1  \\
    \end{array} \right] \\ \\
\end{array}\\
\end{equation*}
\end{minipage}

Proposition~\ref{prop:varswap} allows us to swap two variables, which corresponds to reordering some of the bits in the bit vector. Again, the effect is easiest to see when viewing the bit vector as an $n$-dimensional array, in which case it  corresponds to swapping the axes.  For a formula with two variables, we therefore simply transpose the matrix:

\begin{minipage}{.44\linewidth}
\begin{equation*}
\label{eq:reorder}
\begin{array}{c c}

q(x) :: (x\; y) \rightarrow
\begin{array}{c c c}
    & & x \\
    & & \begin{array}{c c c}
        a & b & c\\
    \end{array}\\
    y &
    \begin{array}{c}
        a \\
        b \\
        c \\
    \end{array} & \left[
    \begin{array}{c c c}
    0 & 1 & 1  \\
    0 & 1 & 1  \\
    0 & 1 & 1  \\
    \end{array} \right] \\
\end{array}\\
\\
\end{array}
\end{equation*}
\end{minipage}
\begin{minipage}{.49\linewidth}
\begin{equation*}
\begin{array}{c c}
q(x) :: (y\; x) \rightarrow
\begin{array}{c c c}
    & & y \\
    & & \begin{array}{c c c}
        a & b & c\\
    \end{array}\\
    x &
    \begin{array}{c}
        a \\
        b \\
        c \\
    \end{array} & \left[
    \begin{array}{c c c}
    0 & 0 & 0  \\
    1 & 1 & 1  \\
    1 & 1 & 1  \\
    \end{array} \right] \\
\end{array}\\

\end{array}
\end{equation*}
\end{minipage}

Combining these operations, we can for instance compute $[(p(y) \land q(x))::(y\;x)]_S = \{(b,b), (b,c)\}$ as $[p(y) :: (x, y)] \land [q(x) :: (x, y)]$:

\[
  \begin{array}{c c c}
    & & y \\
    & & \begin{array}{c c c}
        a & b & c\\
    \end{array}\\
    x &
    \begin{array}{c}
        a \\
        b \\
        c \\
    \end{array} & \left[
    \begin{array}{c c c}
    0 & 1 & 0  \\
    0 & 1 & 0  \\
    0 & 1 & 0  \\
    \end{array} \right] \\ \\
  \end{array}
\quad  \land \quad
\begin{array}{c c c}
    & & y \\
    & & \begin{array}{c c c}
        a & b & c\\
    \end{array}\\
    x &
    \begin{array}{c}
        a \\
        b \\
        c \\
    \end{array} & \left[
    \begin{array}{c c c}
    0 & 0 & 0  \\
    1 & 1 & 1  \\
    1 & 1 & 1  \\
    \end{array} \right] \\ \\
\end{array}
\quad = \quad
\begin{array}{c c c}
    & & y \\
    & & \begin{array}{c c c}
        a & b & c\\
    \end{array}\\
    x &
    \begin{array}{c}
        a \\
        b \\
        c \\
    \end{array} & \left[
    \begin{array}{c c c}
    0 & 0 & 0  \\
    0 & 1 & 0  \\
    0 & 1 & 0  \\
    \end{array} \right] \\ \\
\end{array}
\]

Finally, Proposition~\ref{prop:quant} shows that we can compute the satisfying set for a universal (or existential) quantifier by collapsing the corresponding dimension of our bit matrix using a conjunction (or disjunction) operation.  For instance, if the satisfying set of some formula $F::(x,y)$ is as follows, we can compute the satisfying set of $(\forall y~F)::(x)$ as:

\begin{equation*}
\label{eq:var_elim}
\begin{array}{c c}

F :: (x, y) \rightarrow
\begin{array}{c c}
    & x \\
    y &
    \left[
    \begin{array}{c c c}
    1 & 1 & 1  \\
    1 & 0 & 0  \\
    1 & 1 & 0  \\
    \end{array} \right] \\
\end{array}\hspace{2cm}
(\forall y~F) :: (x) \rightarrow
\begin{array}{c c}
    x \\
    \left[
    \begin{array}{c c c}
    1 & 0 & 0  \\
    \end{array} \right] \\
\end{array}\\

\end{array}
\end{equation*}


A final point of implementation concerns the interpretation of terms.  Consider for instance the formula $f(x) + g(x) = 5$, where $f$ and $g$ are functions (interpreted by $S$, as are all our symbols) that have some interval $[i,j]$ of natural numbers as their codomain.  From a logical perspective, this is an atom in which the equality predicate $=/2$ is applied to the arguments $f(x) + g(x)$ and $5$, both of which are terms.  To apply our approach, we need to represent the satisfying set for this atom by a bit vector of size $\lvert \tau_x^S \rvert$, in which there is a \textit{1} whenever $(f(x)  + g(x))^S = 5$.  One way of constructing this bit vector could be to simply iterate over all elements $d$ of $\tau_x^S$ and compute $(f(x)  + g(x))^{S[x/d]}$ for each one. 
However, it is likely more efficient to represent these functions by vectors of natural numbers directly, instead of bit vectors, to avoid such a costly computation.
In other words, we represent $f$ by the vector $(f^S(d_1)\ldots f^S(d_n))$, where $d1,\ldots, d_n$ is the enumeration of $\tau_x^S$ and $f^S$ is the interpretation of the function $f$ in $S$.  In this way, we can again use the efficient SIMD-parallelism of modern processors to compute the sum vector for $f+g$ from the individual vectors for $f$ and $g$.

\section{Using satisfying sets for grounding}\label{s:moreimplementation}
\label{satset_ground}

So far, we have computed satisfying sets $[\phi]_S$ of formulas under the assumption that $S$ interprets all the symbols in $\phi$.  As explained in Section \ref{s:preliminary}, however, our ultimate goal is to start from a structure $S_0$ that interprets only \emph{some} of the symbols in $\phi$, and to use this information to reduce $\phi$ to an equivalent (in the context of $S_0$) formula $\psi$ that contains only symbols not interpreted in $S$.  

Suppose that $\phi$ is of the form $\forall \vec{x}~\gamma \Rightarrow \psi$ or $\exists \vec{x}~\gamma \land \psi$, such that $S_0$ interprets the symbols in $\gamma$. In this case, we can use the satisfying sets for $\gamma$ in $S_0$ to ground $\phi$ to $\bigwedge_{\vec{d} \in [\gamma::\vec{x}]_{S_0}} \psi[\vec{x} / \vec{d}]$
or  $\bigvee_{\vec{d} \in [\gamma::\vec{x}]_{S_0}} \psi[\vec{x} / \vec{d}]$, respectively.

To transform arbitrary formulas into one of these two forms, we use the following approach. We first identify all maximal subformulas of $\phi$ that contain only symbols interpreted by $S_0$.  Denoting these maximal subformulas as $\gamma_1, \ldots, \gamma_m$, we then construct all $2^m$ possible formulas of the form $\forall \vec{x}~(\bigwedge_{i=1}^m (\lnot) \gamma_i) \Rightarrow \psi'$, in which each $\gamma_i$ appears either negated or not. Here, $\psi'$ is the result of replacing in $\psi$ each subformula $\gamma_i$ by either $\top$ or $\bot$, depending on whether $\gamma_i$ appears negated or not.

\begin{example}
    Let $\Sigma$ be a vocabulary with type $T$ and predicates $P/1$ and $Q/1$,
    for which the arguments are of type $T$.
    Let $S$ be a structure that contains the interpretations of $T$ and $P$.
    Let $F = \forall x: P(x) \lor Q(x)$.
    We can transforming this formula into the form $\forall \vec{x}~\gamma \Rightarrow \psi$.
    This results in the following formula: $(\forall x: P(x) \Rightarrow Q(x) \lor \top) \land (\forall x: \lnot P(x) \Rightarrow Q(x) \lor \bot)$.
    Further simplifying this formula gets us $\forall x: \lnot P(x) \Rightarrow Q(x)$.
    Since $P$ is interpreted in $S$, when grounding this formula we only need to consider values for $x$ such that $\lnot P(x)$ holds.
  \end{example}
  
Even though we can efficiently compute a bit vector representation $\vec{b}$ for $[\gamma::\vec{x}]_{S_0}$, note that to actually compute the grounding, we still have to iterate over all tuples in $\prod_{x_i \in \vec{x}} \tau_{x_i}$ and, for each tuple, check the value of the corresponding bit $b_j \in \vec{b}$.  Therefore, in the worst case, this approach may not actually be faster than Algorithm~\ref{alg:naive}.  Only if the formulas $\gamma$ are sufficiently complex would we expect to see gains.

\section{Related Work}\label{s:relatedwork}

Bit vectors (also known as bitmaps) are used in databases and search engines when applicable for their high speed for set operations~\cite{lemireConsistentlyFasterSmaller2016}. As far as we are aware, however, there is no work that that uses bit vectors to make grounding more efficient in a ground-and-solve architecture. Since the grounding bottleneck is a substantial challenge in current development of logical reasoners, there is quite a lot of literature surrounding efficient grounding.

\textit{Top-down} approaches follow strategies similar to Algorithm \ref{alg:naive}, possibly optimising the way in which formulas are simplified along the way. For instance, the IDP-Z3 reasoning engine~\cite{IDP-Z3} largely follows this approach to grounding.

In the field of Answer Set Programming, most modern grounders utilize various database evaluation techniques~\cite{Kaminski2023} to ground \textit{bottom-up} instead.
This, paired with clever simplification strategies, results in efficient grounding algorithms.
For a comprehensive overview of grounding in ASP, we refer to~\cite{Kaminski2023}.

Another possible approach to reduce the grounding bottleneck is \textit{lazy grounding}~\cite{DalPalu2009,DeCat2015,Leutgeb2018}, in which grounding and solving are interleaved by executing multiple \textit{ground-and-solve} steps, each on a part of the theory at a time.
The general idea is to prevent a blow-up in grounding by iteratively deriving new facts, which can be used to more efficiently ground the next part of the theory.
It may be useful to investigate whether lazy grounding could also integrated into out approach.

In \cite{wittocxGroundingFOID2010}, the authors describe an advanced interpretation/simplification of formulas by computing so-called ``bounds'', which are used to ground formulas w.r.t. redundant (interpreted) symbols.
In essence, their grounder automatically derives symbolic upper and lower bounds for each formula, which it can add to each formula as redundant information.
These bounds allow the grounder to, in some cases, significantly reduce the resulting grounding.
Here too, the approach is quite complimentary to the one presented in this paper, as generating such bounds would allow our algorithm to more efficiently perform the technique outlined in Section~\ref{s:moreimplementation}.

One completely different approach to grounding consists of directly using a relational database system to perform this process~\cite{Augustine2023,Niu2011}.
Here, grounding is expressed as a sequence of SQL queries, which are executed by the database system which in turn fill tables with ground variables representing the grounding.
Crucially, this technique allows leveraging the broad number of query optimization techniques~\cite{Chaudhuri1998} that nowadays come pre-baked in most, if not all, commercially available databases.

\section{Evaluation}\label{s:evaluation}

\begin{table}[]
\caption{(Top) average grounding time in seconds for given benchmark. \\(Bottom) average grounding+solve time, number of timeouts between brackets.}
\label{tab:ground_time}
\centering
\begin{tabular}{lSSSSS}
    \toprule
    Benchmark & {SLI vec} & {SLI naive} & {SLI no reduc} & {gringo} & {IDP-Z3}\\
    \midrule
    CI-SAT          & 0.920 & \B 0.860 & 2.179(30) & 3.744 & {timeout} \\
    CI-UNSAT        & \B 0.681 & 0.806 & 1.661(28) & 2.844 & {timeout} \\
    CS-SAT         & \B 0.774 & 0.893 & 1.771(24) & 3.429 & {timeout} \\
    CS-UNSAT       & 0.765 & \B 0.710 & 1.764(31) & 3.604 & {timeout} \\
    GG           & \B 0.035(2) & 0.045(2) & 0.036(2) & 0.036(3) & 0.232(2) \\
    GC           & 0.052(42) & 0.068(42) & 0.082(43) & \B 0.024(32) & 0.309(42)\\
    PPM& \B 0.069 & 0.097(5) & 0.084 & 7.810(4) & 2.55(1)\\
    WSP & 0.031 & 0.033 & 0.0304 & \B 0.0165 & 0.205\\
    TG & 0.075(4) & 2.06(5) & 3.57(6) & \B 0.0238 & 6.81(6)\\
    \midrule
    \midrule
    CI-SAT & 1.007 & \B 0.947 & 2.279(30) & 4.089 & {timeout} \\
    CI-UNSAT & \B 0.687 & 0.812 & 1.676(28) & 3.108 &  {timeout} \\
    CS-SAT & \B 0.842 & 0.964 & 1.868(24) & 3.886 &  {timeout} \\
    CS-UNSAT & 0.771 & \B 0.716 & 1.778(31) & 3.883 & {timeout}  \\
    GG &  64.12(2) & 63.91(2) & 58.58(2) & 96.25(3) & \B 35.28(2)\\
    GC &  105.0(42) & 105.7(42) & 113.81(43) & 62.54(32) & \B 40.75(42) \\
    PPM & 24.46 & 24.51(5) & \B 18.10 & 41.89(4) & 38.61(1)\\
    WSP &11.44 & 12.21 & 11.40 & \B 10.12 & 19.21 \\
    TG & 0.458(4) & 2.45(5) & 4.09(6) & \B 0.0332 & 8.679(6)\\
    \bottomrule\\
\end{tabular}
\end{table}

The techniques described in Sections~\ref{s:implementation} and~\ref{s:moreimplementation} have been implemented in a new grounder,  named SLI.  Similar to IDP-Z3, it accepts problems formalized in the \fodot language, which it grounds to SMT-LIB that is passed on to the Z3 solver.
SLI is written in Rust, and available under the LGPLv3 license\footnote{\url{https://gitlab.com/EAVISE/sli/SLI}}.

We compare three grounding methods in SLI:
\begin{itemize}
    \item \emph{SLI vec}: The bit vector approach presented in this paper (\textit{SLI vec}).
    \item \emph{SLI naive}:  Algorithm~\ref{alg:naive} to naively compute a reduced grounding (\textit{SLI naive}).
    \item \emph{SLI no reduc}: Algorithm~\ref{alg:naive} without the \emph{if}-test: this computes the full non-reduced grounding and passes it on to Z3, together with the interpretations for the known predicates.
\end{itemize}
We also include clingo v5.7.1 (using \texttt{--mode gringo} to measure grounding) and IDP-Z3 v0.11.1 in our benchmark.
The benchmark is run on a Xeon Silver 4210R CPU with 16 GB of RAM with a timeout of 600 seconds.

As our work concerns grounding, we use a number of benchmarks that require only trivial solving.
For two unary predicates $p$ and $q$, the  \textit{CommonItem (CI)} benchmark checks whether $p$ \textit{and} $q$ holds for a common element (i.e., $\exists x \in \tau: p(x) \land q(x)$) and the \textit{CompleteSet (CS)} benchmarks checks whether $p$ \textit{and} $q$ cover the entire domain (i.e., $\forall x \in \tau: p(x) \lor q(x)$).
We generate random instances for both benchmarks and split them into satisfiable (SAT) and unsatisfiable (UNSAT) instances.
For CS-SAT instances, we consider three levels (0.1, 0.001 and 0.0001) of overlap of between $p$ and $q$, i.e., the ratio of the size of $[p(x) \land q(x)]_S$ to the total size the type $\tau$. (For CS-UNSAT, the overlap ratio is always 0.) Similarly, for CI-UNSAT instances, we consider three levels of non-coverage, i.e., the ratio of $[\neg p(x) \land \neg q(x)]_S$ to the total size of type $\tau$; for CI-SAT instances, this ratio is always 0.

To gauge the effect of sparse interpretations, we also introduce a \textit{Triangle Graph (TG)} benchmark, in which the goal is to identify all triangles in a graph, i.e., triples $(x,y,z)$ such that $\mathit{edge}(x, y) \land \mathit{edge}(y, z) \land \mathit{edge}(z, x)$.  We vary the number of nodes $n$ and generate $n/3$ edges, so the ratio of actual edges to potential edges is $\frac{3}{n}$, i.e., larger graphs grow linearly sparser.

The above benchmarks are intended to evaluate the grounding time in different ways.  However, this is only part of the story, since also grounding quality matters, i.e., whether the grounding can also be efficiently used by the solver. We therefore also include four benchmarks from the ASP2013 competition~\cite{ASPComp2013}, which still require a significant solving effort once the grounding has been produced: Graceful Graph (GG), Graph Coloring (GC), Permutation Pattern Matching (PPM) and Weighted Sequence Problem (WSP).


All benchmarks and scripts are available online\footnote{\url{https://gitlab.com/EAVISE/sli/vectorized_interpretation_benchmark}}.

Table~\ref{tab:ground_time} shows the average grounding time (top) and the average total grounding+solving time (bottom) over all instances of each benchmark.  The average times for CI and CS are  computed from 50 different instances, each having a randomly chosen domain size between 50~000 and 500~000.  For the ASP2013 benchmarks, the averages are computed by running all the instances of each benchmark.
Timeouts are mentioned separately and omitted from the average.

The grounding averages show that the two SLI implementations that compute a reduced grounding (i.e., \emph{SLI vec} and \emph{SLI naive}) outperform the other approaches on all benchmarks apart from \emph{TG}.  Of these two, the bitvector approach \emph{SLI vec} indeed seems faster overall, even though \emph{SLI naive} also wins two benchmarks and the difference is never large.  When we consider the total grounding+solving time, we see that, as expected, the CS and CI benchmarks require very little solving, so the grounding results just carry over. For the ASP competition benchmarks, grounding only accounts for a small portion of the total time.  We can observe here that the SLI+Z3 approach is quite competitive for these benchmarks, indicating that the groundings we produce are indeed suitable for efficient solving.  It is interesting to note that  also \emph{SLI no reduc}  does quite well on these benchmarks, which indicates that for solving-heavy benchmarks, having a larger grounding may actually help Z3 (e.g., by allowing the CDCL algorithm to learn more useful clauses).  


Figure~\ref{fig:setgraphs} analyses the CI and CS benchmarks in more detail. Here, each data point represents the average of 3 different instances of the same size.  For all of these benchmarks, because $p$ and $q$ are provided as input, both of the reducing SLI variants (i.e., \emph{SLI vec} and \emph{SLI naive}) will always produce either $\bot$ or $\top$ as grounding, so there is no further need to invoke a solver.  By contrast, both \emph{SLI no reduc} and \emph{gringo} seem to produce a non-trivial grounding, for which the solver is then invoked.  When comparing \emph{SLI vec} and \emph{SLI naive}, there is clear difference between CS-SAT/CI-UNSAT on the one hand and CS-UNSAT/CI-SAT on the other hand.  This is explained by the fact that in CS-UNSAT and CI-SAT, a single domain element (either one that does not belong to either $p$ or $q$, or one that belongs to both $p$ and $q$, respectively) suffices to immediately conclude that the instance is UNSAT or SAT, respectively.  Our implementation of Algorithm~\ref{alg:naive}  takes advantage of this by exiting the \emph{for}-loop as soon as such an element is found. By contrast, the bit vector implementation always covers the entire domain. As the overlap percentage for CS-SAT increases, so does the probablity of encountering such an element early on in the \emph{for}-loop. 
For CompleteSets-UNSAT we would expect to see the same phenomenon also be visible, but in Figures~\ref{graph:completesetsunsat01} and~\ref{graph:completesetsunsat0001} the difference between \textit{naive} and \textit{vec} is almost non-existent.

A known weakness of bit vectors it that they are bad at representing sparse data.  Figure~\ref{graph:triangle_graph}  shows this weakness is present in the TG benchmarks, where the bit vectors need to reserve a bit position for every one of the $n^3$ \emph{potential} triangles in the graph, even though the number of actual triangles is far lower.  Here, \emph{SLI no reduc} is even worse, since it actually produces a ground formula for each of the $n^3$ possibilities, which both \emph{SLI vec} and \emph{SLI naive} avoid.  Nevertheless, \emph{SLI vec} and \emph{SLI naive} still go through all of the $n^3$ possibilities. The SLI methods also hit the memory limit quite quickly, which we believe is partially due to an implementation detail that can be rectified in the future.  In any case, gringo does much better here, by virtue of its bottom-up approach that only considers the $O(n^2)$ actual edges in the graph.

In the literature, compressed bit vectors \cite{Lemire2018} are typically used to overcome this issue with sparse data.
In future work, we plan to investigate their usage for SLI as well.

\begin{figure}
\centering
\caption{Graph of grounding time for TriangleGraph benchmark}
\label{graph:triangle_graph}
\resizebox {0.4\textwidth}{!}{
\begin{tikzpicture}
    \begin{axis}
        [
            xlabel={amount of nodes},
            ylabel={time [s]},
            ymode=log,
            legend style={fill=none, draw=none},
            legend cell align={left},
            width={0.7\textwidth},
        ]
        \addplot table [x=size, y=totalground, col sep=comma]
            {csv_plots/TriangleGraph/SLI_vectorized.csv};
        \addplot table [x=size, y=totalground, col sep=comma]
            {csv_plots/TriangleGraph/SLI_naive.csv};
        \addplot table [x=size, y=totalground, col sep=comma]
            {csv_plots/TriangleGraph/SLI_no_ground.csv};
        \addplot table [x=size, y=totalground, col sep=comma]
            {csv_plots/TriangleGraph/ASP.csv};
        \addplot table [x=size, y=totalground, col sep=comma]
            {csv_plots/TriangleGraph/IDP-Z3.csv};
        \legend{SLI vec, SLI naive, SLI no reduc, clingo, IDP-Z3}
    \end{axis}
\end{tikzpicture}
}

%
\caption{Graph grounding time for CommonItem and CompleteSet benchmarks}
\label{fig:setgraphs}
\resizebox {0.4\textwidth}{!}{
\subfloat[\centering CompleteSets-SAT]{
\begin{tikzpicture}
    \label{graph:completesetssat}
    \begin{axis}
        [
            xlabel={size},
            ylabel={time [s]},
            legend pos={north west},
            legend style={fill=none, draw=none},
            legend cell align={left},
            width={0.7\textwidth},
        ]
        \addplot table [x=size, y=total, col sep=comma]
            {csv_plots/CompleteSets-Rising-SAT/SLI_vectorized.csv};
        \addplot table [x=size, y=total, col sep=comma]
            {csv_plots/CompleteSets-Rising-SAT/SLI_naive.csv};
        \addplot table [x=size, y=total, col sep=comma]
            {csv_plots/CompleteSets-Rising-SAT/SLI_no_ground.csv};
        \addplot table [x=size, y=total, col sep=comma]
            {csv_plots/CompleteSets-Rising-SAT/ASP.csv};
        \legend{SLI vec, SLI naive, SLI no reduc, clingo, IDP-Z3}
    \end{axis}
\end{tikzpicture}
}}
\resizebox {0.4\textwidth}{!}{
\subfloat[\centering CommonItem-UNSAT]{
\begin{tikzpicture}
    \label{graph:commonitemunsat}
    \begin{axis}
        [
            xlabel={size},
            ylabel={time [s]},
            legend pos={north west},
            legend style={fill=none, draw=none},
            legend cell align={left},
            width={0.7\textwidth},
        ]
        \addplot table [x=size, y=total, col sep=comma]
            {csv_plots/CommonItem-Rising-UNSAT/SLI_vectorized.csv};
        \addplot table [x=size, y=total, col sep=comma]
            {csv_plots/CommonItem-Rising-UNSAT/SLI_naive.csv};
        \addplot table [x=size, y=total, col sep=comma]
            {csv_plots/CommonItem-Rising-UNSAT/SLI_no_ground.csv};
        \addplot table [x=size, y=total, col sep=comma]
            {csv_plots/CommonItem-Rising-UNSAT/ASP.csv};
        \legend{SLI vec, SLI naive, SLI no reduc, clingo, IDP-Z3}
    \end{axis}
\end{tikzpicture}
}}

\resizebox {0.4\textwidth}{!}{
\subfloat[\centering CompleteSets-UNSAT 0.1 non-coverage]{
\begin{tikzpicture}
    \label{graph:completesetsunsat01}
    \begin{axis}
        [
            xlabel={size},
            ylabel={time [s]},
            legend pos={north west},
            legend style={fill=none, draw=none},
            legend cell align={left},
            width={0.7\textwidth},
        ]
        \addplot table [x=size, y=totalground, col sep=comma]
            {csv_plots/CompleteSets-01-UNSAT/SLI_vectorized.csv};
        \addplot table [x=size, y=totalground, col sep=comma]
            {csv_plots/CompleteSets-01-UNSAT/SLI_naive.csv};
        \addplot[color=black,mark=star] table [x=size, y=totalground, col sep=comma]
            {csv_plots/CompleteSets-01-UNSAT/ASP.csv};
        \legend{
            SLI vec, SLI  naive, clingo
        }
    \end{axis}
\end{tikzpicture}
}}
\resizebox {0.4\textwidth}{!}{
\subfloat[\centering CompleteSets-UNSAT 0.001 non-coverage]{
\begin{tikzpicture}
    \label{graph:completesetsunsat0001}
    \begin{axis}
        [
            xlabel={size},
            ylabel={time [s]},
            legend pos={north west},
            legend style={fill=none, draw=none},
            legend cell align={left},
            width={0.7\textwidth},
        ]
        \addplot table [x=size, y=totalground, col sep=comma]
            {csv_plots/CompleteSets-0001-UNSAT/SLI_vectorized.csv};
        \addplot table [x=size, y=totalground, col sep=comma]
            {csv_plots/CompleteSets-0001-UNSAT/SLI_naive.csv};
        \addplot[color=black,mark=star] table [x=size, y=totalground, col sep=comma]
            {csv_plots/CompleteSets-0001-UNSAT/ASP.csv};
        \legend{
            SLI vec, SLI  naive, clingo
        }
    \end{axis}
\end{tikzpicture}
}}
\resizebox {0.4\textwidth}{!}{
\subfloat[\centering CommonItem-SAT 0.1 overlap]{
\begin{tikzpicture}
    \label{graph:commonitemsat01}
    \begin{axis}
        [
            xlabel={size},
            ylabel={time [s]},
            legend pos={north west},
            legend style={fill=none, draw=none},
            legend cell align={left},
            width={0.7\textwidth},
        ]
        \addplot table [x=size, y=total, col sep=comma]
            {csv_plots/CommonItem-01-SAT/SLI_vectorized.csv};
        \addplot table [x=size, y=total, col sep=comma]
            {csv_plots/CommonItem-01-SAT/SLI_naive.csv};
        \addplot[color=black,mark=star] table [x=size, y=total, col sep=comma]
            {csv_plots/CommonItem-01-SAT/ASP.csv};

        \legend{
            SLI vec, SLI  naive, clingo
        }
    \end{axis}
\end{tikzpicture}
}}
\resizebox {0.4\textwidth}{!}{
\subfloat[\centering CommonItem-SAT 0.001 overlap]{
\begin{tikzpicture}
    \label{graph:commonitemsat0001}
    \begin{axis}
        [
            xlabel={size},
            ylabel={time [s]},
            legend pos={north west},
            legend style={fill=none, draw=none},
            legend cell align={left},
            width={0.7\textwidth},
        ]
        \addplot table [x=size, y=totalground, col sep=comma]
            {csv_plots/CommonItem-0001-SAT/SLI_vectorized.csv};
        \addplot table [x=size, y=totalground, col sep=comma]
            {csv_plots/CommonItem-0001-SAT/SLI_naive.csv};
        \addplot[color=black,mark=star] table [x=size, y=totalground, col sep=comma]
            {csv_plots/CommonItem-0001-SAT/ASP.csv};

        \legend{
            SLI vec, SLI naive, clingo
        }
    \end{axis}
\end{tikzpicture}
}}
\end{figure}

\section{Conclusion}\label{s:conclusion}

In this paper, we have described a theoretical method to ground and simplify FO formulas using sets, and have demonstrated a concrete implementation using bit vectors.
Compared to grounding using an enumeration algorithm, this method allows us to delay grounding variables until after a formula has been simplified.
Our experiments show that this approach is competitive w.r.t. state-of-the-art approaches, in part thanks to SIMD parallelism of modern processors.
This confirms that using satisfying sets for simplification could increase a reasoning engine's performance when many symbols are already interpreted.





\section*{Acknowledgements}

This research received funding from the Flemish Government under the “Onderzoeksprogramma Artifici\"ele Intelligentie (AI) Vlaanderen” programme and by Flanders Innovation \& Entrepreneurship (Tetra-project HBC.2022.0071).
The authors would also like to thank Bart Bogaerts and Jo Devriendt.

\bibliographystyle{splncs04}
\bibliography{biblio}

\end{document}